\documentclass[runningheads,a4paper]{llncs}

\usepackage{amssymb}
\usepackage{graphicx}

\usepackage{url}
\urldef{\mailsa}\path|{alfred.hofmann, ursula.barth, ingrid.haas, frank.holzwarth,|
\urldef{\mailsb}\path|anna.kramer, leonie.kunz, christine.reiss, nicole.sator,|
\urldef{\mailsc}\path|erika.siebert-cole, peter.strasser, lncs}@springer.com|

\usepackage{amsfonts}
\usepackage{bbm}
\usepackage{times}
\usepackage{latexsym,amsmath,amsfonts,mathrsfs}
\usepackage{amssymb}
\usepackage{amstext}
\usepackage{amsopn}
\usepackage{cite}
\usepackage[toc,page]{appendix}
\usepackage{tikz}

\usepackage{algorithm}
\usepackage{algorithmic}

\numberwithin{equation}{section}

\begin{document}

\mainmatter
\title{Packing-Based Approximation Algorithm for the $k$-Set Cover Problem}

\titlerunning{Packing-Based Approximation Algorithm for the $k$-Set Cover Problem}

\author{
Martin F\"{u}rer\thanks{Research supported in part by NSF Grant CCF-0728921 and CCF-0964655} \and Huiwen Yu
}
\institute{
Department of Computer Science and Engineering \\
The Pennsylvania State University, University Park, PA 16802, USA
}


\toctitle{Lecture Notes in Computer Science}
\tocauthor{Authors' Instructions}

\date{}
\maketitle

\begin{abstract}
We present a packing-based approximation algorithm for the $k$-Set Cover problem. We introduce a new local search-based $k$-set packing heuristic, and call it Restricted $k$-Set Packing. We analyze its tight approximation ratio via a complicated combinatorial argument. Equipped with the Restricted $k$-Set Packing algorithm, our $k$-Set Cover algorithm is composed of the $k$-Set Packing heuristic \cite{schrijver} for $k\geq 7$, Restricted $k$-Set Packing for $k=6,5,4$ and the semi-local $(2,1)$-improvement \cite{furer} for 3-Set Cover. We show that our algorithm obtains a tight approximation ratio of $H_k-0.6402+\Theta(\frac{1}{k})$, where $H_k$ is the $k$-th harmonic number. For small $k$, our results are $1.8667$ for $k=6$, $1.7333$ for $k=5$ and $1.5208$ for $k=4$. Our algorithm improves the currently best approximation ratio for the $k$-Set Cover problem of any $k\geq 4$.

\end{abstract}


\section{Introduction}

Given a set of elements $U$ and a collection of subsets $\mathscr{S}$ of $U$ with each subset of $\mathscr{S}$ having size at most $k$ and the union of $\mathscr{S}$ being $U$, the $k$-Set Cover problem is to find a minimal size sub-collection of $\mathscr{S}$ whose union remains $U$. Without loss of generality, we assume that $\mathscr{S}$ is closed under subsets. Then the objective of the $k$-Set Cover problem can be viewed as finding a disjoint union of sets of $\mathscr{S}$ which covers $U$.

The $k$-Set Cover problem is NP-hard for any $k\geq 3$. For $k=2$, the 2-Set Cover problem is polynomial-time solvable by a maximum matching algorithm. The greedy approach for approximating the $k$-Set Cover problem chooses a maximal collection of $i$-sets (sets with size $i$) for each $i$ from $k$ down to 1. It achieves a tight approximation ratio $H_k$ (the k-th harmonic number) \cite{johnson}. The hardness result by Feige \cite{feige} shows that for $n=|U|$, the Set Cover problem is not approximable within $(1-\epsilon)\ln n$ for any $\epsilon>0$ unless NP$\subseteq$DTIME($n^{\log\log n}$). For the $k$-Set Cover problem, Trevisan \cite{trevisan} shows that no polynomial-time algorithm has an approximation ratio better than $\ln k-\Omega(\ln\ln k)$ unless subexponential-time deterministic algorithms for NP-hard problems exist. Therefore, it is unlikely that a tremendous improvement of the approximation ratio is possible.

There is no evidence that the $\ln k-\Omega(\ln\ln k)$ lower bound can be achieved. Research on approximating the $k$-Set Cover problem has been focused on improving the positive constant $c$ in the approximation ratio $H_k-c$. Small improvements on the constant might lead us closer to the optimal ratio. One of the main ideas based on greedy algorithms is to handle small sets separately. Goldschmidt et al. \cite{gold} give a heuristic using a matching computation to deal with sets of size 2 and obtain an $H_k-\frac{1}{6}$ approximation ratio. Halld\'{o}rsson \cite{hall2} improves $c$ to $\frac{1}{3}$ via his ``t-change''  and ``augmenting path'' techniques. Duh and F\"{u}rer \cite{furer} give a semi-local search algorithm for the 3-Set Cover problem and further improve $c$ to $\frac{1}{2}$. They also present a tight example for their semi-local search algorithm.

A different idea is to replace the greedy approach by a set-packing approach. Levin \cite{levin} uses a set-packing algorithm for packing 4-sets and improves $c$ to $0.5026$ for $k\geq 4$. Athanassopoulos et al. \cite{lp} substitute the greedy phases for $k\geq 6$ with packing phases and reach an approximation ratio $H_k-0.5902$ for $k\geq 6$.

The goal of this paper is not to provide incremental improvement in the approximation ratio for $k$-Set Cover. We rather want to obtain the best such result achievable by current methods. It might be the best possible result, as we conjecture the lower bound presented in \cite{trevisan} not to be optimal.

In this paper, we give a complete packing-based approximation algorithm (in short, PRPSLI) for the $k$-Set Cover problem. For $k\geq 7$, we use the $k$-set packing heuristic introduced by Hurkens and Shrijver \cite{schrijver}, which achieves the best known to date approximation ratio $\frac{2}{k}-\epsilon$ for the $k$-Set Packing problem for any $\epsilon>0$. On the other hand, the best hardness result by Hazan et al. \cite{hazan} shows that it is NP-hard to approximate the $k$-Set Packing problem within $\Omega(\frac{\ln k}{k})$.

For $k=6,5,4$, we use the same packing heuristic with the restriction that any local improvement should not increase the number of 1-sets which are needed to finish the disjoint set cover. We call this new heuristic Restricted $k$-Set Packing. We prove that for any $k\geq 5$, the Restricted $k$-Set Packing algorithm achieves the same approximation ratio as the corresponding unrestricted set packing heuristic. For $k=4$, this is not the case. The approximation ratio of the Restricted 4-Set Packing algorithm is $\frac{7}{16}$, which is worse than the $\frac{1}{2}-\epsilon$ ratio of the 4-set packing heuristic but it is also tight. For $k=3$, we use the semi-local optimization technique \cite{furer}. We thereby obtain the currently best approximation ratio for the $k$-Set Cover problem. Table 1 (in Appendix Section 5) includes a comparison of the approximation ratio of our algorithm with $GRSLI_{k,5}$ \cite{furer}, Levin's algorithm \cite{levin} and $PRSLI_{k,5}$ \cite{lp}. We also show that our result is indeed tight. Thus, $k$-Set Cover algorithms which are based on packing heuristic can hardly be improved. Our novel Restricted $k$-Set Packing algorithm is quite simple and natural, but its analysis is complicated. It
is essentially based on combinatorial arguments. We use the factor-revealing linear programming analysis for the $k$-Set Cover problem. The factor-revealing linear program is introduced by Jain et al. \cite{lp1} for analyzing the facility location problem. Athanassopoulos et al. \cite{lp} are the first to apply it to the $k$-Set Cover problem.

The paper is organized as follows. In Section 2, we give the description of our algorithm and present the main results. In Section 3, we prove the approximation ratio of the Restricted $k$-Set Packing algorithm. In Section 4, we analyze our $k$-Set Cover algorithm via the factor-revealing linear program.

\section{Algorithm Description and the Main Theorem}

In this section, we describe our packing-based $k$-Set Cover approximation algorithm. We first give an overview of some existing results.

Duh and F\"{u}rer \cite{furer} introduce a semi-local $(s,t)$-improvement for the 3-Set Cover problem. First, it greedily selects a maximal disjoint union of 3-sets. Then each local improvement replaces $t$ 3-sets with $s$ 3-sets, if and only if after computing a maximum matching of the remaining elements, either the total number of sets in the cover decreases, or it remains the same, while the number of 1-sets decreases. They also show that the $(2,1)$-improvement algorithm gives the best performance ratio for the $3$-Set Cover problem among all semi-local $(s,t)$-improvement algorithms. The ratio is proved to be tight.

\begin{theorem}[\cite{furer}]
The semi-local $(2,1)$-optimization algorithm for 3-Set Cover produces a solution with performance ratio $\frac{4}{3}$. It uses a minimal number of 1-sets.
\end{theorem}

We use the semi-local $(2,1)$-improvement as the basis of our $k$-Set Cover algorithm. Other phases of the algorithm are based on the set packing heuristic \cite{schrijver}. For fixed $s$, the heuristic starts with an arbitrary maximal packing, it replaces $p\leq s$ sets in the packing with $p+1$ sets if the resulting collection is still a packing. Hurkens and Shrijver \cite{schrijver} show the following result,

\begin{theorem}[\cite{schrijver}]
For all $\epsilon>0$, the local search $k$-Set Packing algorithm for parameter $s=O(\log_k \frac{1}{\epsilon})$ has an approximation ratio $\frac{2}{k}-\epsilon$.
\end{theorem}

The worst-case ratio is also known to be tight. We apply this packing heuristic for $k\geq 7$. For $k=6,5,4$, we follow the intuition of the semi-local improvement and modify the local search of the packing heuristic, requiring that any improvement does not increase the number of 1-sets. We use the semi-local (2,1)-improvement for 3-Set Cover to compute the number of 1-sets required to finish the cover. Lemma 2.2 in \cite{furer} guarantees that the number of 1-sets returned by the semi-local (2,1)-improvement is no more than this number in any optimal solution. We compute this number first at the beginning of the restricted phase. Each time we want to make a replacement via the packing heuristic, we compute the number of 1-sets needed to finish the cover after making the replacement. If this number increases, the replacement is prohibited. To summarize, we call our algorithm the Restricted Packing-based $k$-Set Cover algorithm (PRPSLI) and give the pseudo-code in Algorithm 1. For input parameter $\epsilon>0$, $s_i$ is the parameter of the local improvement in Phase $i$. For any $i\neq 5,6$, we set $s_i$ in the same way as in Theorem 2. For $i=5,6$, we set $s_i=\lceil\frac{2}{i\epsilon}\rceil$.

\begin{algorithm}
\caption{Packing-based $k$-Set Cover Algorithm (PRPSLI)}
\begin{algorithmic}

\STATE \COMMENT{\textbf{\emph{ The $k$-Set Packing Phase}}}
\FOR{$i\leftarrow k$ down to 7} \STATE Select a maximal collection of disjoint $i$-sets.\REPEAT \STATE Select $p\leq s_i$ $i$-sets and replace them with $p+1$ $i$-sets. \UNTIL{there exist no more such improvements.} \ENDFOR

\STATE \COMMENT{\textbf{\emph{ The Restricted $k$-Set Packing Phase}}}
\STATE Run the semi-local $(2,1)$-improvement algorithm for 3-Set Cover on the remaining uncovered elements to obtain the number of 1-sets.
\FOR{$i\leftarrow 6$ to 4}
\REPEAT \STATE Try to replace $p\leq s_i$ $i$-sets with $p+1$ $i$-sets. Commit to the replacement only if the number of 1-sets computed by the semi-local $(2,1)$-improvement algorithm for 3-Set Cover on the remaining uncovered elements does not increase. \UNTIL{there exist no more such improvements.} \ENDFOR

\STATE \COMMENT{\textbf{\emph{ The Semi-Local Optimization Phase}}}
\STATE Run the semi-local $(2,1)$-improvement algorithm on the remaining uncovered elements.

\end{algorithmic}
\end{algorithm}

The algorithm clearly runs in polynomial time. The approximation ratio of PRPSLI is presented in the following main theorem. For completeness, we also state the approximation ratio for the 3-Set Cover problem, which is obtained by Duh and F\"{u}rer \cite{furer} and remains the best result. Let $\rho_k$ be the approximation ratio of the $k$-Set Cover problem.

\begin{theorem}[Main]
For all $\epsilon >0$, the Packing-based $k$-Set Cover algorithm has an approximation ratio $\rho_k=2H_k-H_{\frac{k}{2}}+\frac{2}{k}-\frac{1}{k-1}-\frac{4}{3}+\epsilon$ for even $k$ and $k\geq 6$; $\rho_k=2H_k-H_{\frac{k-1}{2}}-\frac{4}{3}+\epsilon$ for odd $k$ and $k\geq 7$; $\rho_5=1.7333$;
$\rho_4=1.5208$; $\rho_3=\frac{4}{3}$.
\end{theorem}

\begin{remark}
For odd $k\geq 7$, the approximation ratio $\rho_k$ is derived from the expression $\rho_k=\frac{2}{k}+\cdots+\frac{2}{5}+\frac{1}{3}+1+\epsilon$. We can further obtain the asymptotic representation of $\rho_k$, i.e., $\rho_k=2H_k-H_{\frac{k-1}{2}}-\frac{4}{3}+\epsilon=H_k+\ln 2-\frac{4}{3}+\Theta(\frac{1}{k})+\epsilon=H_k-0.6402+\Theta(\frac{1}{k})+\epsilon$. Similarly, for even $k\geq 6$, $\rho_k=\frac{2}{k}+\frac{1}{k-1}+\frac{2}{k-3}+\cdots+\frac{2}{5}+\frac{1}{3}+1+\epsilon=2H_k-H_{\frac{k}{2}}+\frac{2}{k}-\frac{1}{k-1}-\frac{4}{3}+\epsilon=H_k+\ln 2-\frac{4}{3}+\Theta(\frac{1}{k})+\epsilon=H_k-0.6402+\Theta(\frac{1}{k})+\epsilon$.
Finally, $\rho_{5} = \frac{2}{5} + \frac{1}{3} + 1$ and $\rho_{4} =  \frac{7}{16} + \frac{1}{12} + 1$.
\end{remark}

\begin{remark}
Restriction on Phase 6 is only required for obtaining the approximation ratio $\rho_k$ for even $k$ and $k\leq 12$. In other cases, only restriction on Phase 5 and Phase 4 are necessary.
\end{remark}


We prove the main theorem in Section 4. Before that, we analyze the approximation ratio of the Restricted $k$-Set Packing algorithm for $k\geq 4$ in Section 3. We state the result of the approximation ratio of the Restricted $k$-Set Packing algorithm as follows.

\begin{theorem}[Restricted $k$-Set Packing]
There exists a Restricted 4-Set Packing algorithm which has an approximation ratio $\frac{7}{16}$. For all $\epsilon>0$ and for any $k\geq 5$, there exists a Restricted $k$-Set Packing algorithm which has an approximation ratio $\frac{2}{k}-\epsilon$.
\end{theorem}

\begin{remark}
Without loss of generality, we assume that optimal solution of the Restricted $k$-Set Packing problem also has the property that it does not increase the number of 1-sets needed to finish the cover of the remaining uncovered elements. This assumption is justified by Lemma 2 in Appendix 9.1.
\end{remark}

\section{The Restricted $k$-Set Packing Algorithm}

We fix one optimal solution $\mathscr{O}$ of the Restricted $k$-Set Packing algorithm. We refer to the sets in $\mathscr{O}$ as optimal sets. For fixed $s$, a local improvement replaces $p\leq s$ $k$-sets with $p+1$ $k$-sets. We pick a packing of $k$-sets $\mathscr{A}$ that cannot be improved by the Restricted $k$-Set Packing algorithm. We say an optimal set is an \emph{$i$-level set} if exactly $i$ of its elements are covered by sets in $\mathscr{A}$. For the sake of analysis, we call a local improvement an \emph{$i$-$j$-improvement} if it replaces $i$ sets in $\mathscr{A}$ with $j$ sets in $\mathscr{O}$. As a convention in the rest of the paper, small letters represent elements, capital letters represent subsets of $U$, and calligraphic letters represent collections of sets. We first introduce the notion of blocking.

\subsection{Blocking}

The main difference between unrestricted $k$-set packing and restricted $k$-set packing is the restriction on the number of 1-sets which are needed to finish the covering via the semi-local (2,1)-improvement. This restriction can prohibit a local improvement. If any $i$-$j$-improvement is prohibited because of an increase of 1-sets, we say there exists a $blocking$. In Example 1 given in Appendix Section 6.1, we construct an instance of 4-set packing to help explain how blocking works.

We now define blocking formally. We are given a fixed optimal $k$-set packing $\mathscr{O}$ of $U$ and a $k$-set packing $\mathscr{A}$ chosen by the Restricted $k$-Set Packing algorithm. We consider all possible extensions of $\mathscr{A}$ to a disjoint cover of $U$ by 1-sets, 2-sets and 3-sets. We order these extensions lexicographically, first by the number of 1-sets, second by the total number of 2-sets and 3-sets which are not within a $k$-set of $\mathscr{O}$, and third by the number of 3-sets which are not within a $k$-set of $\mathscr{O}$. We are interested in the lexicographically first extension. Notice that we pick this specific extension for analysis only. We cannot obtain this ordering without access to $\mathscr{O}$. We explain how we order the extensions in Example 2 (Appendix Section 6.2).

Suppose we finish the cover from the packing $\mathscr{A}$ with the lexicographically first extension. Let $\mathscr{F}$ be an undirected graph such that each vertex in $\mathscr{F}$ represents an optimal set. Two vertices are adjacent if and only if there is a 2-set in the extension intersecting with the corresponding optimal sets. Since the number of 2-sets and 3-sets not within an optimal set is minimized, there are no multiple edges in the graph. For brevity, when we talk about a node $V$ in $\mathscr{F}$, we also refer to $V$ as the corresponding optimal set. Moreover, when we say the \emph{degree} of a node $V$, we refer to the number of neighbors of $V$.

\begin{proposition}
$\mathscr{F}$ is a forest.
\end{proposition}

\begin{proposition}
For any $i<k-1$, there is no 1-set inside an $i$-level set. i.e. 1-set can only appear in $(k-1)$-level sets.
\end{proposition}

\begin{proposition}
For any tree $\mathscr{T}$ in $\mathscr{F}$, there is at most one node which represents an $i$-level set, such that the degree of the node is smaller than $k-i$.
\end{proposition}

For any tree, if there exists a node with property in Proposition 3, we define it to be the root. Otherwise, we know that all degree 1 nodes represent $(k-1)$-level sets. We define an arbitrary node not representing a $(k-1)$-level set to be the root. If there are only $(k-1)$-level sets in the tree, i.e. the tree degenerates to one edge or a single point, we define an arbitrary $(k-1)$-level set to be the root. All leaves represent $(k-1)$-level sets. (The root is not considered to by a leaf.) We call such a tree a \emph{blocking tree}. For any subtree, we say that the leaves \emph{block} the nodes in this subtree. We also call the set represented by a leaf a \emph{blocking set}.

We consider one further property of the root.

\begin{proposition}
Let $k \geq 4$.
In any blocking tree, there exists at most one node of either 0-level or 1-level that is of degree 2. If such a node exists, it is the root.
\end{proposition}

The proofs of Proposition 1 to 4 are given in Appendix Section 6.3.

Based on these simple structures of the blocking tree, we are now ready to prove the approximation ratio of the Restricted $k$-Set Packing algorithm.

\subsection{Analysis of the Restricted 4-Set Packing Algorithm}

We prove in this section that the Restricted 4-Set Packing algorithm has an approximation ratio $\frac{7}{16}$. We first explain how this $\frac{7}{16}$ ratio is derived. We use the unit $\mathcal{U}$ defined in Example 1 (Appendix Section 6.1). Assume when the algorithm stops, we have $n\gg 1$ copies of $\mathcal{U}$ and a relatively small number of 3-level sets. We denote the $i$-th copy of $\mathcal{U}$ by $\mathcal{U}_i$. For each $i$ and $1\leq j\leq 12$, the set $O_j$ in $\mathcal{U}_i$ and $\mathcal{U}_{i+1}$ are adjacent. This chain of $O_j$'s starts from and ends at a 3-level set respectively. Then the performance ratio of this instance is slightly larger than $\frac{7}{16}$. We first prove that the approximation ratio of the Restricted 4-Set Packing algorithm is at least $\frac{7}{16}$.

Given $\mathscr{F}$, a collection of blocking trees. We assign 4 tokens to every element covered by sets chosen by the restricted packing algorithm. We say a set has a free token if after distributing the token, this set retains at least 7 tokens. We show that we can always distribute the tokens among all the optimal sets $\mathscr{O}$, so that there are at least 7 tokens in each optimal set.

\begin{proof}
We present the first round of redistribution. \\


\textbf{Round 1 - Redistribution in each blocking tree $\mathscr{T}$}. Every leaf in $\mathscr{T}$ has 4 free tokens to distribute. Every internal node $V$ of degree $d$ requests $4(d-2)$ tokens from a leaf. We consider each node with nonzero request in the reverse order given by breadth first search (BFS).

        \begin{itemize}
            \item If $d=3$, $V$ requests 4 tokens from any leaf in the subtree rooted at $V$ which has 12 tokens.
            \item If $d=4$, $V$ has three children $V_1,V_2,V_3$. $V$ sends requests of 4 tokens to any leaf in the subtree rooted at $V_1,V_2,V_3$, one for each subtree. $V$ takes any two donations of 4 tokens.
            \item The root of degree $r$ receives the rest of the tokens contributed by the leaves.
        \end{itemize}

\begin{proposition}
After Round 1, every internal node in $\mathscr{T}$ has at least 8 tokens, the root of degree $r$ has $4r$ tokens.
\label{prop1-k4}
\end{proposition}

Proposition 5 is proved in Appendix Section 7.1. According to Proposition 4, we know that after the first round of redistribution every node has at least 8 tokens except any 0-level roots which are of degree 1 and any singletons in $\mathscr{F}$ which are 1-level sets. \\

We first consider the collection of 1-level sets $\mathscr{S}_1$ which are singletons in $\mathscr{F}$. Let $S_1$ be such a 1-level set that intersects with a 4-set $A$ chosen by the algorithm. Assume $A$ also intersects with $j$ other optimal sets $\{O_i\}_{i=1}^j$.

We point out that no $O_i$ belongs to $\mathscr{S}_1$. Otherwise suppose $O_i\in\mathscr{S}_1$. Then there is a 1-2-improvement (replace $A$ with $S_1$ and $O_i$).

We give the second round of redistribution, such that after this round, every set in $\mathscr{S}_1$ has at least 7 tokens.
For optimal sets $O,W$, consider each token request sent to $O$ from $W$. We say it is an \emph{internal request}, if $W\in\mathscr{S}_1$, and $W$ and $O$ intersect with a set $A$ chosen by the algorithm. Otherwise, we say it is an \emph{external request}. \\

\textbf{Round 2 - Redistribution for $S_1\in\mathscr{S}_1$}. $S_1$ sends $\tau$ requests of one token to $O_i$ if $|O_i\bigcap A|=\tau$. For each node $V$ in $\mathscr{T}$, internal requests are considered prior to external requests.

        For each request sent to $V$ from $W$,
        \begin{itemize}
            \item If $V$ has at least 8 tokens, give one to $W$.
            \item If $V$ has only 7 tokens. \\
            (1) If $V$ is a leaf, it requests from the node which has received 4 tokens from it during the first round of redistribution.\\
            (2) If $V$ is not a leaf,

            (2.1) If $W\in\mathscr{S}_1$, $V$ requests a token from a leaf which has at least 8 tokens in the subtree rooted at $V$.  

            (2.2) If $W$ is a node in $\mathscr{T}$. Suppose $V$ has children $V_1,..,V_d$ and $W$ belongs to the subtree rooted at $V_1$. $V$ requests from a leaf  which has at least 8 tokens in the subtree rooted at $V_2,...,V_d$.
        \end{itemize}

We now prove the correctness of the second round of redistribution.

\begin{proposition}
Every singleton node $O$ of level $j$ has at most $j-1$ requests. Every leaf has at most $k-2$ internal requests. Every internal node of level $j$ has at most $j$ requests. The root of level $s$ has at most $s$ requests.
\end{proposition}

\begin{proposition}
Singleton node $O$ can satisfy all the requests.
\end{proposition}

\begin{proposition}
The root $R$ of level $s$ and degree $r$ can satisfy all the requests if $s+r\geq 2$.
\end{proposition}

\begin{proposition}
There is no external request sent to a root which is of level 0 and degree 1.
\end{proposition}

\begin{proposition}
Any external request sent from a leaf $L$ can be satisfied.
\end{proposition}

\begin{proposition}
Any external request sent from an internal node of degree $d\geq 3$ can be satisfied.
\end{proposition}

\begin{proposition}
Any external request sent from an internal node $V$ of degree 2 can be satisfied.
\end{proposition}

The proofs of Proposition 6 to 12 are given in Appendix Section 7.1. From Proposition 7, 8, 10, 11 and 12, we know that all requests can be satisfied. Hence after the second round of distribution, every set in $\mathscr{S}_1$ has at least 7 tokens. And from Proposition 9, we know that every root of level 0 and degree 1 retains 4 tokens. \\

We consider a root $R$ which is a 0-level set of degree 1. $R$ receives 4 tokens from leaf $B$. Assume $B$ is covered by $\{A_i\}_{i=1}^j\in\mathscr{A}$ and $\{A_i\}_{i=1}^j$ intersect with $\{O_i\}_{i=1}^l\in\mathscr{O}$. We first prove that,

\begin{proposition}
$\forall O\in\{O_i\}_{i=1}^l$, $O$ does not receive any token request during Round 2 of redistribution.
\end{proposition}

The proof of Proposition 13 is given in Appendix Section 7.1. Based on Proposition 13, for any set $O\in\{O_i\}_{i=1}^l$, we can think of a token request from $R$ to $B$ as an internal request to $O$. We describe the third round of redistribution.

\textbf{Round 3 - Redistribution for root of level 0 and degree 1}. Request 1 token from each of $O_1,...,O_l$ following Round 2.

\begin{proposition}
$l\geq 3$.
\end{proposition}

The correctness of Round 3 follows from Proposition 14. The proof of Proposition 14 is given in Appendix Section 7.1. We thus prove that a root of level 0 and degree 1 has 7 tokens after the third round of redistribution.


Therefore, after the three rounds of token redistribution, each optimal set has at least 7 tokens, then the approximation ratio of the Restricted 4-Set Packing algorithm is at least $\frac{7}{16}$. \qed
\end{proof}

We give the construction of tight example in Appendix Section 7.2. We thus conclude that the approximation ratio of the Restricted 4-Set Packing algorithm is $\frac{7}{16}$.

\subsection{Analysis of the Restricted $k$-Set Packing Algorithm, $k\geq 5$}

For $k\geq 5$, we prove that the approximation ratio of the Restricted $k$-Set Packing algorithm is the same as the set packing heuristic \cite{schrijver}. For fixed $s$, a local improvement can replace at most $s$ sets with $s+1$ sets. We prove that for any $\epsilon>0$, there exists an $s$, such that the approximation ratio of the Restricted $k$-Set Packing algorithm is at least $\frac{2}{k}-\epsilon$.

The proof strategy is similar to that of the Restricted 4-Set Packing algorithm. We first create a forest of blocking trees $\mathscr{F}$. We give every element covered by sets chosen by the algorithm one unit of tokens. We then redistribute the tokens among all the optimal sets. We claim that for $s\geq\frac{2}{k\epsilon}$, after redistribution, every optimal set gets at least $2-k\epsilon$ units of tokens. We use a different parameter of local improvement from Theorem 2. The algorithm still runs in polynomial time.

\begin{proof}
The first round of redistribution goes as follows. \\

\textbf{Round 1 - Redistribution in each blocking tree $\mathscr{T}$}. Every leaf in $\mathscr{T}$ has $k-3$ units of free tokens to distribute. Every internal node $V$ of degree $d$ receives $d-2$ units of tokens from a leaf. The root receives the remaining tokens.

\begin{proposition}
After Round 1, every node has at least 2 units of tokens, except singletons which are 1-level sets.
\end{proposition}

We consider the collection of 1-level sets $\mathscr{S}_1$ which are singletons in $\mathscr{F}$. Let $S_1$ be such a 1-level set that intersects with a $k$-set $A$ chosen by the algorithm. Assume $A$ intersects with $j$ other optimal sets $\{Q_i\}_{i=1}^j$. \\

\textbf{Round 2 - Redistribution for $S_1\in\mathscr{S}_1$}.
    \begin{itemize}
    \item For every singleton node of level $i$ ($i\geq 3$):

    Send 1 unit of tokens each to arbitrarily $i-2$ internal requests.
    \item If $S_1$ receives one unit of tokens, we are done. Otherwise, pick an arbitrary singleton node $Q_i$ from $\{Q_i\}_{i=1}^j$. Let $Q_i=O_1$. Let $A_1\in\mathscr{A}$ be a set intersecting with $Q_i$ while it does not intersect with any set in $\mathscr{S}_1$. (The existence of $A_1$ follows from Proposition 6.)
    \item If $A_1$ intersects with some singleton node which has at least 3 units of tokens, we move 1 unit to $S_1$. Otherwise pick an arbitrary singleton node $O_2$ which intersects with $A_1$. Let $A_2\in\mathscr{A}$ be a set intersecting with $O_2$ while it does not intersect with any set in $\mathscr{S}_1$.
    \item Repeat this procedure and form a chain of singleton sets $O_1=Q_i,O_2,...,O_p$ such that $S_1$ and $O_1$ intersect with $A=A_0$, $O_i$ and $O_{i+1}$ intersect with a set $A_i$ chosen by the algorithm, for $i=1,..,p-1$, until

        (1) The chain ends when excluding $O_p$, every set intersecting with $A_{p-1}$ is not a singleton node. Denote the collection of such $S_1$ by $\mathscr{S}_1^T$.

        Move 1 unit of tokens from any node in the collection of non-degenerate blocking trees which has at least 3 units of tokens to $S_1$.

        (2) The chain ends where there exists a 1-level set $S_1'\in\mathscr{S}_1$ which intersects with $A'$, such that $S_1'$ starts another chain of length $q$ which ends at $O_p(=O_q)$ (as illustrated in Fig. 4 (Appendix Section 8.1)).

        Construct a graph $\mathscr{G}$, such that every vertex in the graph represents a node in a chain, $V_1,V_2$ are connected, if the corresponding nodes intersect with a set $A$ chosen by the algorithm. Consider every connected component $\mathscr{C}$ in $\mathscr{G}$. Equally distribute all tokens to every vertex in $\mathscr{C}$.



    \end{itemize}

The correctness of step (1) and step (2) follows from Proposition 16 and 17. We give the proofs in Appendix Section 8.2.
\begin{proposition}
The collection of non-degenerate blocking trees has at least $|\mathscr{S}_1^T|$ free units of tokens. After step (1), every set in $\mathscr{S}^T$ has 2 units of tokens.
\end{proposition}

\begin{proposition}
After step (2), every optimal set has at least $2-k\epsilon$ units of tokens.
\end{proposition}

We conclude that after the second round of redistribution, every optimal set gets at least $2-k\epsilon$ units of tokens. Therefore, the approximation ratio of the Restricted $k$-Set Packing algorithm is at least $\frac{2}{k}-\epsilon$. \qed
\end{proof}

Moreover, the tight example of the $k$-set packing heuristic \cite{schrijver} can also serve as a tight example of the Restricted $k$-Set Packing algorithm for $k\geq 5$. Hence, the approximation ratio of the Restricted $k$-Set Packing algorithm is $\frac{2}{k}-\epsilon$. \\


\section{Analysis of the Algorithm PRPSLI}

We use the factor-revealing linear program introduced by Jain et al. \cite{lp1} to analyze the approximation ratio of the algorithm PRPSLI. Athanassopoulos et al. \cite{lp} first apply this method to the $k$-Set Cover problem. Notice that the cover produced by the restricted set packing algorithms is a cover which minimizes the number of 1-sets. In Appendix Section 9.1, we first show that for any $k\geq 4$, there exists a $k$-set cover which simultaneously minimizes the size of the cover and the number of 1-sets in the cover. We then present the set-ups and notations of the factor-revealing linear program (LP).

The proof of Theorem 3 is similar as the proof of Theorem 6 in \cite{lp}. Namely, we find a feasible solution to the dual program of (LP), which makes the objective function of the dual program equal to the value of $\rho_k$ defined in Theorem 3, thus $\rho_k$ is an upper bound of the approximation ratio of PRPSLI. We give the proof of an upper bound of the approximation ratio of PRPSLI in Appendix Section 9.2. On the other side, we give an instance for each $k$, such that PRPSLI does not achieve a better ratio than $\rho_k$ on this instance in Appendix Section 9.3. We thus prove Theorem 3.\\

\bibliographystyle{plain}

\bibliography{ISAACbib}

\newpage


\section*{Appendix}

\section{Comparison on the Approximation Ratio of the $k$-Set Cover Problem with Previous Works}

\begin{table}
\caption{Comparison on the Approximation Ratio of the $k$-Set Cover Problem}
\begin{center}
\begin{tabular}{|c|c|c|c|c|}

  \hline
  $k$ & $GRSLI_{k,5}$\cite{furer} & \cite{levin} & $PRSLI_{k,5}$\cite{lp} & PRPSLI \\
  \hline
  3 & 1.3333 & 1.3333 & 1.3333 & 1.3333 \\
  4 & 1.5833 & 1.5808 & 1.5833 & 1.5208 \\
  5 & 1.7833 & 1.7801 & 1.7833 & 1.7333 \\
  6 & 1.9500 & 1.9474 & 1.9208 & 1.8667 \\
  7 & 2.0929 & 2.0903 & 2.0690 & 2.0190 \\
  8 & 2.2179 & 2.2153 & 2.1762 & 2.1262 \\
  9 & 2.3290 & 2.3264 & 2.2917 & 2.2413 \\
  10 & 2.4290 & 2.4264 & 2.3802 & 2.3302 \\
  20 & 3.0977 & 3.0952 & 3.0305 & 2.9779\\
  21 & 3.1454 & 3.1428 & 3.0784 & 3.0284\\
  50 & 3.9992 & 3.9966 & 3.9187 & 3.8683\\
  75 & 4.4014 & 4.3988 & 4.3178 & 4.2678\\
  100 & 4.6874 & 4.6848 & 4.6021 & 4.5520 \\
  \hline
  large $k$ & $H_k-0.5$ & $H_k-0.5026$ & $H_k-0.5902$ & $H_k-0.6402$ \\
  \hline
\end{tabular}
\end{center}
\end{table}

\section{Section 3.1}

\subsection{Example 1}

\begin{example}[Blocking]
Consider an instance $(U,\mathscr{S})$ of the 4-Set Cover problem. Suppose there is an optimal solution $\mathscr{O}$ of the Restricted 4-Set Packing algorithm which consists of only disjoint 4-sets that cover all elements. Let $\mathscr{O}=\{O_i\}_{i=1}^{16}\bigcup\{B_i\}_{i=1}^m$, $m>12$. Let $\{A_i\}_{i=1}^7$ be a collection of 4-sets chosen by the algorithm. $\{O_i\}_{i=1}^{16}$ is a collection of 1-level or 2-level sets. Denote the $j$-th element of $O_i$ by $o_i^j$, for $j=1,2,3,4$. If $A_i=(o_{i_1}^{j_1},o_{i_2}^{j_2},o_{i_3}^{j_3},o_{i_4}^{j_4})$, we say that $A_i$ covers the elements $o_{i_1}^{j_1},o_{i_2}^{j_2},o_{i_3}^{j_3}$ and $o_{i_4}^{j_4}$. Denote the following unit by $\mathcal{U}$,

\begin{displaymath}
\mathcal{U}=\left\{
  \begin{array}{ll}
    A_1=(o_{1}^{1},o_{5}^{1},o_{9}^{1},o_{13}^{1})\\
    A_2=(o_{2}^{1},o_{6}^{1},o_{10}^{1},o_{14}^{1})\\
    A_3=(o_{3}^{1},o_{7}^{1},o_{11}^{1},o_{15}^{1})\\
    A_4=(o_{4}^{1},o_{8}^{1},o_{12}^{1},o_{16}^{1})\\
    A_5=(o_{1}^{2},o_{2}^{2},o_{3}^{2},o_{4}^{2})\\
    A_6=(o_{5}^{2},o_{6}^{2},o_{7}^{2},o_{8}^{2})\\
    A_7=(o_{9}^{2},o_{10}^{2},o_{11}^{2},o_{12}^{2})
  \end{array}
\right.
\end{displaymath}

We visualize this construction in Fig. 1. Let each cube represent a 4-set in $\{A_i\}_{i=1}^7$ and we place $\{O_i\}_{i=1}^{16}$ vertically within a $4\times 4$ square (not shown in the figure), such that each $O_i$ intersects with one or two sets in $\{A_i\}_{i=1}^{7}$. $\{A_i\}_{i=1}^{7}$ are placed horizontally. Notice that $O_{13},O_{14},O_{15},O_{16}$ which intersects with $A_1,A_2,A_3,A_4$ respectively are 1-level sets. The other 12 sets in $\{O_i\}_{i=1}^{16}$ are 2-level sets.

$\{B_i\}_{i=1}^m$ is a collection of 3-level sets. Notice that for our fixed optimal solution, all elements can be covered by 4-sets, so there is no 1-set needed to finish the cover. For given $\mathscr{S}$, when we compute an extension of the packing to a full cover via the semi-local (2,1)-improvement, assume the unpacked element of $B_i$ ($1\leq i\leq 12$) can only be covered by a 2-set intersecting with both $B_i$ and $O_i$, or it introduces a 1-set in the cover. The remaining unpacked elements of $\{O_i\}_{i=1}^{16}$ and $\{B_i\}_{i=1}^m$ can be covered arbitrarily by 2-sets and 3-sets. In unrestricted packing, one of the local improvements consists of replacing $A_1,A_2,A_3,A_4,A_5$ by $O_1,O_2,O_3,O_4,O_{13},O_{14},O_{15},O_{16}$. However, in restricted packing, for $1\leq i\leq 12$, adding any $O_i$ to the packing would create a 1-set covering the unpacked element of $B_i$ during the semi-local (2,1)-improvement. Hence this local improvement is prohibited as a result of restricting on the number of 1-sets.

We remark that blocking can be much more complicated than in this simple example. As we shall see later in Section 3.2, for the Restricted 4-Set Packing problem, a 3-level set can initiate a blocking of many optimal sets.

\begin{figure}
\begin{center}

\begin{tikzpicture}[scale =0.6]

  \draw[yslant=-0.5] (0,0) rectangle (4,1);
  \draw[yslant=-0.5] (0,1) rectangle (1,2);
  \node at (0.5, 1.3) {$A_5$};
  \draw[yslant=-0.5] (1,1) rectangle (2,2);
  \node at (1.5, 0.8) {$A_6$};
  \draw[yslant=-0.5] (2,1) rectangle (3,2);
  \node at (2.5, 0.3) {$A_7$};

  \draw[yslant=0.5] (4,-4) grid (8,-3);
  \node at (4.5,-1.3) {$A_1$};
  \node at (5.5, -0.8) {$A_2$};
  \node at (6.5, -0.3) {$A_3$};
  \node at (7.5, 0.2) {$A_4$};
  \pgftransformyshift{1.5cm}
  \pgftransformxshift{-1cm}
  \draw[yslant=0.5] (4,-4) rectangle (8,-3);
  \pgftransformyshift{-1.5cm}
  \pgftransformxshift{1cm}

  \pgftransformyshift{-2cm}
  \draw[yslant=0.5,xslant=-1] (4,4) rectangle (8,3);
  \draw[yslant=0.5,xslant=-1] (4,3) rectangle (8,2);
  \draw[yslant=0.5,xslant=-1] (4,2) rectangle (8,1);
  \pgftransformyshift{-1cm}
  \draw[yslant=0.5,xslant=-1] (4,1) grid (8,0);

\end{tikzpicture}
\caption{Placement of $A_1$ to $A_7$}
\end{center}
\end{figure}
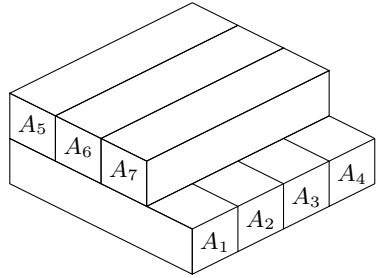

\end{example}

\subsection{Example 2}
\begin{example}[Finish the cover by 1-sets, 2-sets and 3-sets]
In Fig. 2, Fig. 3 and Fig. 4. Rectangles placed vertically represent optimal sets. Circles represent 1-sets. Ellipses placed horizontally represent 2-sets or 3-sets, where the smaller ones stand for 2-sets and the larger ones stand for 3-sets. The cross symbol represents an element covered by a $k$-set in $\mathscr{A}$. Let $(n_1,n_2,n_3)$ be an ordered pair, such that, $n_1$ is the number of 1-sets, $n_2$ is the total number of 2-sets and 3-sets which are not within one optimal set, and $n_3$ is the number of 3-sets which are not within one optimal set. These 1-sets, 2-sets and 3-sets are used to finish the cover. The right picture is always a cover which is before the cover in the left picture in the lexicographic order.
\end{example}

\begin{figure}
\begin{center}

\begin{tikzpicture}[scale =0.6]
    \draw (0,0) rectangle (1,4);
    \draw (1.2,0) rectangle (2.2,4);
    \draw (2.4,0) rectangle (3.4,4);

    \draw (0.5,0.5) node {X} (0.5,1.5) node {X} (0.5,2.5) node {X} (1.7,0.5) node {X} (1.7,1.5) node {X} (1.7,2.5) node {X} (2.9,0.5) node {X} (2.9,1.5) node {X};

    \draw (0.5, 3.5) circle (0.4);
    \draw (2.3,3.5) ellipse (1 and 0.4);
    \draw (2.9, 2.5) circle (0.4);

    \node at (5.8, 2) {$\Longrightarrow$};
    \pgftransformxshift{8cm}
    \draw (0,0) rectangle (1,4);
    \draw (1.2,0) rectangle (2.2,4);
    \draw (2.4,0) rectangle (3.4,4);

    \draw (0.5,0.5) node {X} (0.5,1.5) node {X} (0.5,2.5) node {X} (1.7,0.5) node {X} (1.7,1.5) node {X} (1.7,2.5) node {X} (2.9,0.5) node {X} (2.9,1.5) node {X};

    \draw (0.5, 3.5) circle (0.4);
    \draw (2.9,3) ellipse (0.4 and 0.9);
    \draw (1.7, 3.5) circle (0.4);

\end{tikzpicture}

\caption{$(2,1,0)\Rightarrow (2,0,0)$}
\end{center}
\end{figure}
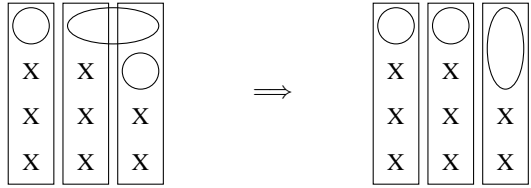

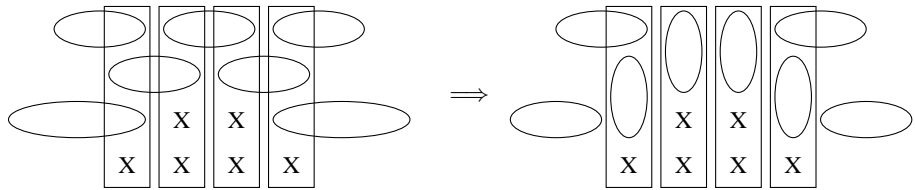
\begin{figure}
\begin{center}

\begin{tikzpicture}[scale =0.6]
    \draw (0,0) rectangle (1,4);
    \draw (1.2,0) rectangle (2.2,4);
    \draw (2.4,0) rectangle (3.4,4);
    \draw (3.6,0) rectangle (4.6,4);
    \draw (0.5,0.5) node {X} (1.7,0.5) node {X} (1.7,1.5) node {X} (2.9,0.5) node {X} (2.9,1.5) node {X} (4.1,0.5) node {X};
    \draw (-0.1,3.5) ellipse (1 and 0.4);
    \draw (1.1, 2.5) ellipse (1 and 0.4);
    \draw (-0.6, 1.5) ellipse (1.5 and 0.4);
    \draw (2.3,3.5) ellipse (1 and 0.4);
    \draw (3.5, 2.5) ellipse (1 and 0.4);
    \draw (4.7, 3.5) ellipse (1 and 0.4);
    \draw (5.2, 1.5) ellipse (1.5 and 0.4);

    \node at (8, 2) {$\Longrightarrow$};
    \pgftransformxshift{11cm}
    \draw (0,0) rectangle (1,4);
    \draw (1.2,0) rectangle (2.2,4);
    \draw (2.4,0) rectangle (3.4,4);
    \draw (3.6,0) rectangle (4.6,4);
    \draw (0.5,0.5) node {X} (1.7,0.5) node {X} (1.7,1.5) node {X} (2.9,0.5) node {X} (2.9,1.5) node {X} (4.1,0.5) node {X};
    \draw (-0.1,3.5) ellipse (1 and 0.4);
    \draw (0.5, 2) ellipse (0.4 and 0.9);
    \draw (-1.1, 1.5) ellipse (1 and 0.4);
    \draw (1.7,3) ellipse (0.4 and 0.9);
    \draw (2.9,3) ellipse (0.4 and 0.9);
    \draw (4.7, 3.5) ellipse (1 and 0.4);
    \draw (4.1,2) ellipse (0.4 and 0.9);
    \draw (5.7, 1.5) ellipse (1 and 0.4);
\end{tikzpicture}

\caption{$(0,7,2)\Rightarrow (0,4,0)$}
\end{center}
\end{figure}

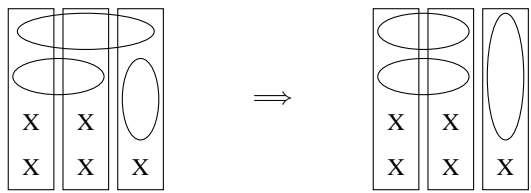
\begin{figure}
\begin{center}
\begin{tikzpicture}[scale=0.6]

    \draw (1.2,0) rectangle (2.2,4);
    \draw (2.4,0) rectangle (3.4,4);
    \draw (3.6,0) rectangle (4.6,4);
    \draw (1.7,0.5) node {X} (1.7,1.5) node {X} (2.9,0.5) node {X} (2.9,1.5) node {X} (4.1,0.5) node {X};
    \draw (2.9,3.5) ellipse (1.5 and 0.4);
    \draw (2.3, 2.5) ellipse (1 and 0.4);
    \draw (4.1, 2) ellipse (0.4 and 0.9);

    \node at (7, 2) {$\Longrightarrow$};
    \pgftransformxshift{8cm}
    \draw (1.2,0) rectangle (2.2,4);
    \draw (2.4,0) rectangle (3.4,4);
    \draw (3.6,0) rectangle (4.6,4);
    \draw (1.7,0.5) node {X} (1.7,1.5) node {X} (2.9,0.5) node {X} (2.9,1.5) node {X} (4.1,0.5) node {X};
    \draw (2.3,3.5) ellipse (1 and 0.4);
    \draw (2.3, 2.5) ellipse (1 and 0.4);
    \draw (4.1, 2.5) ellipse (0.4 and 1.4);

\end{tikzpicture}
\caption{$(0,2,1)\Rightarrow (0,2,0)$}
\end{center}
\end{figure}

\subsection{Proofs in Section 3.1}

\begin{proof}[Proposition 1]
It is sufficient to show that there is no cycle in $\mathscr{F}$. Suppose $V_1,V_2,...,V_l$ form a cycle in $\mathscr{F}$. We remove the 2-sets intersecting with adjacent nodes in the cycle and add the 2-sets inside each $V_i$, for $i=1,...,l$. Then the total number of 2-sets and 3-sets not within an optimal set decreases. \qed
\end{proof}

\begin{proof}[Proposition 2]
Let $V_1$ be an $i$-level set which is covered by a 1-set, for $i<k-1$. If there are more than one 1-set covering $V_1$, we can replace them with 2-set or 3-set inside $V_1$. Hence we only consider the case that there is only one 1-set covering $V_1$. If there is a 2-set or 3-set inside $V_1$, we can remove the 1-set and replace the 2-set or 3-set with a 3-set or two 2-sets respectively. If $V_1$ is a singleton node, since $i<k-1$, there exists a 2-set or 3-set inside $V_1$. Otherwise let $V_l$ be any node of degree 1 which connects to $V_1$ via a simple path $V_1,V_2,...,V_{l-1},V_l$. We remove all 2-sets intersecting with adjacent nodes in this path and move the 1-set from $V_1$ to $V_l$, then add 2-sets inside $V_1,...,V_{l-1}$. If there is a 1-set in any $V_i$, $2\leq i\leq l-1$, it can be replaced together with the 2-set just added inside $V_i$ with a 3-set. If there is another 1-set in $V_l$, it can be combined with the 1-set moved from $V_1$ to a 2-set. Hence, there is no 1-set in any $i$-level set for $i<k-1$. \\ However, 1-set can remain in $(k-1)$-level set. In this case, the $(k-1)$-level set is a singleton node in $\mathscr{F}$. \qed
\end{proof}

\begin{proof}[Proposition 3]
We consider non-degenerate tree. Suppose there are two nodes, an $i$-level set $V_1$ and a $j$-level set $V_2$ which satisfy the requirements. Consider a simple path connecting $V_1$ and $V_2$. We remove all 2-sets intersecting with adjacent nodes on this simple path. For those nodes excluding $V_1,V_2$ on the path, we can add a 2-set inside each optimal set. For $V_1$, since $V_1$ is not a $(k-1)$-level set, there is no 1-set inside $V_1$. Moreover, the degree of $V_1$ is smaller than $k-i$, hence there is a 2-set or 3-set inside $V_1$, we can then replace it with a 3-set or two 2-sets inside $V_1$ respectively. The same argument applies to $V_2$. Hence, there can be at most one $i$-level set with degree smaller than $k-i$. \qed
\end{proof}

\begin{proof}[Proposition 4]
We first prove that there is at most one node of either 0-level or 1-level which is of degree 2. Assume $V_1$ and $V_2$ are 0-level or 1-level sets and of degree 2. We replace the 2-sets intersecting with two adjacent nodes along the simple path connecting $V_1$ and $V_2$ with 2-sets or 3-sets inside the optimal sets represented by these nodes, then the total number of 2-sets and 3-sets not within an optimal set decreases. In this way, $V_1$ and $V_2$ turn to be of degree 1.

Suppose we have one 0-level or 1-level set $V_1$ of degree 2. We prove that there is no node of degree 1 representing an $i$-level set for any $i\leq k-2$, thus $V_1$ is the root. Assume $W$ is an $i$-level set of degree 1. We replace the 2-sets intersecting with two adjacent nodes along the simple path connecting $V_1$ and $W$ with 2-sets or 3-sets inside the optimal sets represented by these nodes, then the total number of 2-sets and 3-sets not within an optimal set decreases. In this way, $V_1$ turns to be of degree 1 and $W$ turns to be of degree 0. \qed
\end{proof}

\section{Section 3.2}

\subsection{Proofs in Section 3.2}

\begin{proof}[Proposition 5]
Consider any subtree with root $V$ of degree $d+1$, $d\geq 2$ ($V$ is not the root of $\mathscr{T}$). Since any internal node of degree 2 does not request any token, for simplicity we assume there is no internal node of degree 2 in the subtree. Assume $V$ has children $V_1,...,V_d$. If every child of $V$ is a leaf, since the request is processed in the reverse order given by BFS, we know that there is no other node requesting tokens from $V_1,...,V_d$, hence $V$'s requests can be satisfied. Assume that in every subtree rooted at $V_1,...,V_d$, all the requests have been satisfied. Since for any tree with root of degree $r$, the quantity
\begin{equation}
\sum_{\textrm{V is an internal node}}(d_V-2)+r \; .
\label{eqn:node}
\end{equation}
equals to the number of the leaves of the tree. We know that there are 4 free tokens in each subtree rooted at $V_1,..,V_d$. Hence, $V$'s requests can be satisfied.

By (\ref{eqn:node}), the root of $\mathscr{T}$ receives $4r$ tokens. \qed
\end{proof}

\begin{proof}[Proposition 6]
Suppose a singleton node $O$ is covered by $\{A_i\}_{i=1}^l$. If every $A_i$ intersects with $S_i\in\mathscr{S}_1$, we have a local improvement by replacing $A_1,...,A_l$ with $S_1,...,S_l$ and $O$. Hence, $\{A_i\}_{i=1}^l$ intersect with at most $l-1\leq j-1$ sets in $\mathscr{S}_1$. $O$ has at most $j-1$ requests of tokens. Similarly, we know that a leaf is a $(k-1)$-level set, it has at most $k-2$ internal requests.

Every internal node of level $j$ can have $j$ requests, and the root of level $s$ can have $s$ requests. \qed
\end{proof}

\begin{proof}[Proposition 7]
Assume $O$ is of level $j$ and $j\geq 2$. From Proposition 6, $O$ has at most $j-1$ request. After giving $j-1$ tokens, $O$ has $3j+1\geq 7$ tokens left. Hence, $O$ can satisfy all the requests. \qed
\end{proof}

\begin{proof}[Proposition 8]
The root has $4(s+r)$ tokens and it has at most $s+r$ requests. If it has at most $s+r-1$ requests, it retains $3s+3r+1\geq 7$ tokens after satisfying all the requests.

It remains to prove that it cannot have $s+r$ requests. Otherwise, suppose $R$ is covered by $\{A_i\}_{i=1}^l$ and every $A_i$ intersects with $S_i\in\mathscr{S}_1$. There are leaves $L_1,...,L_r$ which request token from $R$ for 1-level singleton $S_1',...,S_r'$. $S_i'$ intersects with $A_i'\in\mathscr{A}$. Then we have a $(l+r)$-$(l+r+1)$-improvement (replace $S_1,...,S_l,S_1',...,S_r'$ with $A_1,...,A_l,A_1',...,A_r'$ and $R$). \qed
\end{proof}

\begin{proof}[Proposition 9]
If such a root has an external request from $V$, then $V$ is a leaf and $V$ has an internal request from an $S\in\mathscr{S}_1$, $S$ intersects with an $A\in\mathscr{A}$. Then we have a 1-2-improvement (replace $A$ with $S$ and the root). \qed
\end{proof}

\begin{proof}[Proposition 10]
By Proposition 6, $L$ makes an external request if and only if it has 2 internal requests. As in Step (1), $L$ sends request to $V$ which has received 4 tokens from it. If $V$ has at least 8 tokens, it gives one to $L$. Otherwise, we proceed with Step (2.2). Hence, it is sufficient to prove that in (2.2), there exists a leaf in some subtree rooted at $V_2,...,V_d$ which has 8 tokens.

Otherwise, suppose any leaf in the subtree rooted at $V_2,...,V_d$ has an internal request. We pick $L_2,..,L_d$ belonging to the subtree rooted at $V_2,..,V_d$ respectively. $L_i$ has an internal request from $S_i$ which intersects with $A_i\in\mathscr{A}$. Moreover, we know that $L$ has an internal request from $S$ intersecting with $A\in\mathscr{A}$. We have a $d$-$(d+1)$-improvement (replace $A_2,...,A_d,A$ with $S_2,...,S_d,S$ and $V$). \qed
\end{proof}

\begin{proof}[Proposition 11]
An internal node of degree at least 3 sends a request if and only if it receives a request from a leaf but fails to satisfy it. The proof is indeed contained in the proof of Proposition 10. \qed
\end{proof}

\begin{proof}[Proposition 12]
If $V$ makes an external request, by Proposition 6, $V$ has two internal requests, say from $S_1,S_2$. $S_1$ intersects with $A_1\in\mathscr{A}$, $S_2$ intersects with $A_2\in\mathscr{A}$.

Let $X$ be a closest node to $V$ which belongs to the subtree rooted at $X$ and has a degree $d\geq 3$. If no such $X$ exists, let $X=V$ and $d=2$. If on the contrary, $V$'s request cannot be satisfied, we pick a leaf in each subtree rooted at a child of $X$. Let $L_1,...,L_{d-1}$ be these leaves, then $L_i$ has an internal request from $S_i'$ intersecting with $A_i'\in\mathscr{A}$. We then have a local improvement by replacing $A_1',...,A_{d-1}',A_1,A_2$ with $S_1',...,S_{d-1}'$, $S_1,S_2$ and $V$. \qed
\end{proof}

\begin{proof}[Proposition 13]
First, $\forall O\notin\mathscr{S}_1$. Otherwise, assume $O$ intersects with $A\in\mathscr{A}$, we have a 1-2-improvement (replace $A$ with $O$ and $R$). Hence, $O$ does not have any internal request.

$O$ does not have any external request either. Recall in Round 2 of redistribution, a leaf has an external request if there exists an internal node $V$ of degree $d$, such that $O$ belongs to the subtree rooted at $V$, and \\
- If $d=2$, $V$ has two internal requests from $S_1,S_2$, $S_1,S_2$ intersect with $A_1,A_2\in\mathscr{A}$ respectively. However in this case, we have a 3-4-improvement (replace $A,A_1,A_2$ with $R,S_1,S_2,V$).\\
- If $d=3$, there is an external request sent to $V$ from leaf $L_1$, $L_1$ has an internal request from $S_1$ which intersects with $A_1\in\mathscr{A}$, and there is an internal request sent to $V$ from $S_2$. Let $S_2$ intersect with $A_2\in\mathscr{A}$. We have a 3-4-improvement (replace $A,A_1,A_2$ with $R,S_1,S_2,V$). \\
- If $d=4$, there are two external requests sent to $V$ from leaves $L_1,L_2$. $L_1,L_2$ have internal requests from $S_1,S_2$. $S_1,S_2$ intersect with $A_1,A_2\in\mathscr{A}$ respectively. There is a 3-4-improvement (replace $A,A_1,A_2$ with $R,S_1,S_2,V$).

Hence, $O$ does not have any token request during Round 2 of redistribution. \qed
\end{proof}

\begin{proof}[Proposition 14]
First, $j\geq 2$. Otherwise, we have a 1-2-improvement (replace $A_1$ with $B$ and $R$).  \\
If $j=2$, then the number of elements contained in
$\{O_i\}_{i=1}^l\bigcap\{A_i\}_{i=1}^j$ is $2\cdot 4-3=5$. Hence $l\geq 2$. If $l=2$, there exists an $O_i$ which is a 3-level set and completely covered by $\{A_i\}_{i=1}^j$. However in this case, we have a 2-3-improvement (replace $A_1,A_2$ with $O_i,B,R$). Hence, $l\geq 3$. \\
If $j=3$, then the number of elements contained in
$\{O_i\}_{i=1}^l\bigcap\{A_i\}_{i=1}^j$ is $3\cdot 4-3=9$. In this case, $l\geq 3$.  \qed
\end{proof}

\subsection{Tight example of the Restricted 4-Set Packing algorithm}

We construct an example showing that for any fixed $s$, which is the parameter of the local improvement, and any $\epsilon>0$ and there exists an instance, on which the Restricted 4-Set Packing algorithm has a performance ratio at most $\frac{7}{16}+\epsilon$. We thereby conclude that the approximation ratio of the Restricted 4-Set Packing algorithm is $\frac{7}{16}$. \\

Our construction is randomized. We take $\frac{n}{12}$ copies of the unit $\mathcal{U}$ defined in Example 1. Then there are $n$ 2-level sets. Suppose there are $m$ 3-level sets which start the blocking. In case of 4-set packing, blocking can start from a single 3-level set $S_0$, then propagate through an arbitrary number of 2-level sets $S_1,S_2,...,S_i$ ($i\geq 1$). More specifically, we form the blocking by covering the remaining uncovered elements of $S_0,S_1,...,S_i$ by 2-sets $T_1,T_2,...,T_i$, such that $T_j$ intersects $S_{j-1}$ and $S_{j}$, for $j=1,2,...,i$.

In our example, we assign each 2-level set to one of the $m$ blocking sets independently and uniformly at random. If a 2-level set is assigned to a blocking set, that blocking set is a leaf of the blocking tree containing the 2-level set. For fixed $s$, we consider local $p$-$(p+1)$-improvements for any $p\leq s$. For each subset of $\mathscr{A}$ with size $p$, suppose after removing one covering set of $\lambda$ blocking sets, we can replace these $p$ sets with at most $q$ optimal sets. Then $q\leq \frac{16p}{7}\leq\frac{16s}{7}$. If $q-p-\lambda\geq 1$, i.e., $\lambda\leq q-p-1\leq\frac{9s}{7}\leq 2s$, there is a local improvement. \\

We assume $n\gg s$. Let $t$ be the maximum number of 2-level sets which can be added to the solution by the local improvement. Then $t\leq \frac{16s}{7}$. Let $\mathscr{O}_2$ be the collection of 2-level sets. Let $\{\mathcal{E}_1,\mathcal{E}_2,...,\mathcal{E}_N\}$ be a set of random variables, where $N=\sum_{i=1}^t{n\choose i}\approx n^{t}$ is the number of nonempty subsets of $\mathscr{O}_2$ with size at most $t$. Assume we enumerate every 2-level sets and arrange all subsets of $\mathscr{O}_2$ lexicographically. Let $\mathcal{E}_i$ be the event that there exists a local improvement which adds the $i$-th subset with size at most $t$ of $\mathscr{O}_2$ to the solution. Let $Y$ be the number of blocking sets assigned to the 2-level sets in this subset. Then,

\begin{equation}
\Pr(\mathcal{E}_i)\leq \Pr(Y\leq 2s)=\sum_{\lambda=1}^{2s} \Pr(Y=\lambda) \; .
\end{equation}

Since the assignments of 2-level sets to blocking sets are independently and uniformly at random, we bound (7.2) as follows,

\begin{eqnarray}
\Pr(\mathcal{E}_i) &\leq& \sum_{\lambda=1}^{2s}{m\choose\lambda}(\frac{\lambda}{m})^N \leq \sum_{\lambda=1}^{2s} \frac{m^{\lambda}\lambda^N}{m^N} \nonumber \\
&\leq& (2s)^N\cdot m^{-(N-2s)}\approx (\frac{2s}{m})^{N} \; .
\end{eqnarray}

Assume that $m\gg s$. Since each $\mathcal{E}_i$ depends on less than $N\approx n^t$ elements in $\{\mathcal{E}_i\}_{i=1}^N$, and $N\cdot \Pr(\mathcal{E}_i)=o(1)$. By the following lemma, we have $\Pr(\bigcap_{i=1}^N \overline{\mathcal{E}_i})>0$.

\begin{lemma}[Corollary 5.12 \cite{random}]
Let $\{\mathcal{E}_1,\mathcal{E}_2,...,\mathcal{E}_n\}$ be events in a probability space with $\Pr(\mathcal{E}_i)\leq p$ for all $i$. If each event is mutually independent of all other events except for at most $d$, and $ep(d+1)\leq 1$, then $\Pr(\bigcap_{i=1}^n \overline{\mathcal{E}_i})>0$.
\end{lemma}

Therefore, there exists an assignment of 2-level sets to blocking sets such that no local $p$-$(p+1)$-improvement is possible, for $p\leq s$. \\

Assume all 3-level sets but a constant few of them are blocking sets. The performance ratio of the Restricted 4-Set Packing algorithm on this instance is $\frac{7n+9m}{16n+12m}=\frac{7}{16}+O(\frac{1}{n})$. The ratio tends arbitrarily close to $\frac{7}{16}$ when $n\rightarrow \infty$.

\section{Section 3.3}

\subsection{Figure 5}

\begin{figure}
\begin{center}
\begin{tikzpicture}[scale=0.6]

    \draw (0,0) rectangle (0.8,5); 
    \node at (0.4, 4.5) {$S_1$};
    \draw (2,0) rectangle (2.8,5); 
    \node at (2.4,4.5) {$O_1$};
    \draw (4,0) rectangle (4.8,5); 
    \node at (4.4,4.5) {$O_2$};
    \draw (1.4, 0.5) ellipse (1.3 and 0.4); 
    \node at (1.4,0.5) {$A$};
    \draw (3.4,1.5) ellipse (1.3 and 0.4); 
    \node at (3.4, 1.5) {$A_1$};
    \draw (5.4, 0.5) ellipse (1.3 and 0.4); 
    \node at (5.4, 0.5) {$A_2$};

    \node at (6.5, 2.5) {$\cdots\cdots$};

    \pgftransformxshift{2cm};
    \draw (7,0) rectangle (7.8,5); 
    \node at (7,4.5) {$O_p=O_q$};
    \draw (9,0) rectangle (9.8,5); 
    \node at (9.4,4.5) {$O_{q-1}'$};

    \node at (11.5,2.5) {$\cdots\cdots$};
    \draw (13,0) rectangle (13.8,5); 
    \node at (13.4, 4.5) {$O_1'$};
    \draw (15,0) rectangle (15.8,5); 
    \node at (15.4,4.5) {$S_1'$};

    \draw (6.4, 1.5) ellipse (1.3 and 0.4); 
    \node at (6.4,1.5) {$A_{p-1}$};
    \node at (8.4,0.5) {$A_{p-1}$};
    \draw (8.4,0.5) ellipse (1.3 and 0.4);
    \draw (12.4, 1.5) ellipse (1.3 and 0.4);
    \node at (12.4,1.5) {$A_{1}'$};
    \draw (14.4,0.5) ellipse (1.3 and 0.4); 
    \node at (14.4, 0.5) {$A'$};

\end{tikzpicture}
\caption{Rectangles represent optimal $k$-sets. Ellipses represent $k$-sets chosen by the algorithm.}
\end{center}
\end{figure}
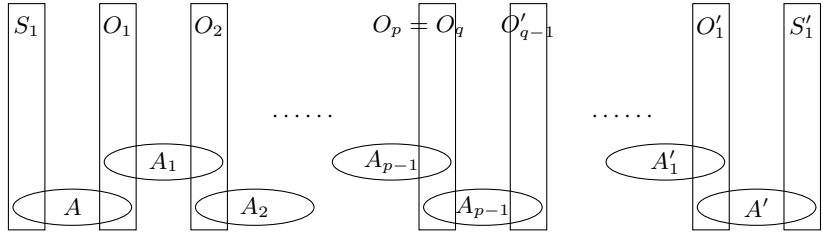

\subsection{Proofs in Section 3.3}

\begin{proof}[Proposition 15]
According to Proposition 4 and (\ref{eqn:node}), we know that after Round 1, every node has at least 2 units of tokens except singletons which are 1-level sets. Note that contrary to the Restricted 4-Set Packing problem, in restricted $k$-set packing for $k\geq 5$, every leaf contributes at least 2 units of tokens. Hence, even if a root is of level 0 and degree 1, it can still receives at least 2 units of tokens after the first round of redistribution.

We remark that we do not care about from which leaf an internal node or the root receive the tokens. \qed
\end{proof}

\begin{proof}[Proposition 16]
Let $N_1=|\mathscr{S}_1^T|$. Without loss of generality, assume there is only one non-degenerate tree $\mathscr{T}$. Let $\mathscr{T}$ have an $s$-level root of degree $r$ and a set of internal nodes with degree set $\{d_V\}_{V\textrm{ is an internal node}}$.

From (\ref{eqn:node}), we know that the leaves contribute $(k-3)\sum (d_V-2)+(k-3)r$ units of tokens. Here the summation goes through the set of internal nodes. There are $\sum (d_V-2+k-d_V-2)+(k-4)\sum (d_V-2)+(k-3)r+s-2$, i.e. $(k-4)\sum (d_V-1)+(k-3)r+s-2$ free units of tokens in $\mathscr{T}$ which can be distributed to sets in $\mathscr{S}_1^T$. \\
Assume on the contrary that
        \begin{equation}
        N_1>(k-4)\sum (d_V-1)+(k-3)r+s-2 \; .
        \label{N1lb}
        \end{equation}

        We derive an upper bound on $N_1$.

        We first claim that any $|A_i\bigcap O_j|\leq 1$. If $|A_i\bigcap O_j|\geq 2$, we know that $O_j$ has at least 3 units of tokens.

        If $A_{p-1}=A$, the $k-1$ elements of $A$ which do not intersect with $\mathscr{S}_1$, intersect with $\mathscr{T}$. If $A_{p-1}\neq A$, for each $S_1$, there are $k-1$ corresponding elements intersecting with $\mathscr{T}$. Here, these $k-1$ elements belong to $A_{p-1}$. Hence, $\mathscr{T}$ covers at least $(k-1)N_1$ elements. On the other hand, there are $(k-1)[\sum (d_V-2)+r]+\sum (k-d_V)+s$ elements covered by $\mathscr{T}$. We have,
        \begin{equation}
        (k-1)N_1\leq (k-1)[\sum (d_V-2)+r]+\sum (k-d_V)+s \; .
        \label{N1ub}
        \end{equation}
        Combining (\ref{N1lb}) and (\ref{N1ub}), we get $(k-1)[(k-4)\sum (d_V-1)+(k-3)r+s-2]<(k-1)[\sum (d_V-2)+r]+\sum (k-d_V)+s$, which implies
        \begin{equation}
        \sum (k^2-6k+6)(d_V-1)+(k^2-5k+5)r+(k-2)s-2(k-1)<0 \; .
        \label{dv1}
        \end{equation}
        (\ref{dv1}) only holds for the case that $k=5,r=1,s=0$ and there are at most 2 internal nodes. In this case, we observe that the leaf cannot intersect with any set $A$ which intersects with $S\in\mathscr{S}_1$, or we have a 1-2-improvement by replacing $A$ with $S$ and the root. Hence, we modify (\ref{N1ub}) and get
        \begin{equation}
        (k-1)N_1< \sum (k-d_V)+s \; .
        \label{N1ub2}
        \end{equation}
        Combining (\ref{N1ub2}) and (\ref{N1lb}), we have $\sum (5d_V-9)<0$. Since $d_V\geq 2$, it leads to a contradiction. \qed
\end{proof}

\begin{proof}[Proposition 17]
Let $e$ be the number of edges and $v$ be the number of vertices in $\mathscr{C}$.

If $\mathscr{C}$ contains a circle, then $e\geq v$. The average number of tokens in $\mathscr{C}$ is $\frac{2e}{v}\geq 2$.

If $\mathscr{C}$ is a tree, then $e=v-1$. We claim that $e\geq s+1$. Otherwise, there exists a $e$-$v$-improvement (replace all the sets corresponding to edges in $\mathscr{C}$ with all the sets corresponding to vertices in $\mathscr{C}$). The average number of tokens in $\mathscr{C}$ is $\frac{2e}{v}\geq 2-\frac{2}{v}\geq 2-\frac{2}{s+2}\geq 2-k\epsilon$ for $s\geq \frac{2}{k\epsilon}$.

Hence, in any case, after collecting all tokens of $\mathscr{C}$ then equally distributed among every vertex in $\mathscr{C}$, each vertex gets at least $2-k\epsilon$ units of tokens.  \qed
\end{proof}

\section{Proof of Theorem 3}

\subsection{Set-up of the Factor-revealing LP}

We first prove that there exists a $k$-set cover which simultaneous minimizes the size and the number of 1-sets in the cover.
\begin{lemma}
Suppose there are $k$-set covers $\mathscr{C}$, $\mathscr{C}'$, where $\mathscr{C}$ has $b$ sets and $b_1$ 1-sets, and $\mathscr{C}'$ has $b'$ sets and $b_1'$ 1-sets. Then there exists a $k$-set cover $\mathscr{C}''$ which has $\min(b, b')$ sets and $\min(b_1,b_1')$ 1-sets.
\end{lemma}

\begin{proof}
We create a bipartite graph $G=(V, V', E)$, where every vertex in $V$($V'$) represents a set in $\mathscr{C}$($\mathscr{C}'$ ), and every edge in $E$ represents an element in the universe. Hence, the degree of a vertex represents the size of the corresponding set. Two vertices being adjacent means the corresponding two sets covers a same element. We show how to find $\mathscr{C}''$.
 
For simplicity, assume $G$ is connected and $b\leq b'$. If $b_1\leq b_1'$, take $\mathscr{C}''$ to be $\mathscr{C}$. Otherwise, consider a vertex $v_1\in V$ of degree 1 such that its neighbor has degree at least 2. If there exists a vertex $v_l\in V$ of degree at least 3, assume $v_l$ is the one with the shortest distance to $v_1$. Consider the path $v_1,v_1',...,v_{l-1}',v_l$ connecting $v_1$ and $v_l$, we replace the corresponding sets of $v_1,...,v_{l-1}$ with $v_1',...,v_{l-1}'$ and delete the element in $v_l$ which is covered by both $v_l$ and $v_{l-1}'$. If any $v_i'$ ($1\leq i\leq l-1$) has degree at least 3, delete the elements in $v_i'$ which are not covered by $v_1,...,v_{l-1}$. In this way, $b_1$ decreases by 1 while $b$ remains the same. If there is no $v_l\in V$ which has degree at least 3, since $b_1>b_1'$, and $\mathscr{C}$ and $\mathscr{C}'$ cover the same universe of elements, $b$ must be greater than $b'$, a contradiction. Hence, we can eventually decrease $b_1$ to be at most $b_1'$. Finally, we take $\mathscr{C}''$ to be this modified $\mathscr{C}$. \qed

\end{proof}

We are now ready to set-up the factor-revealing linear program for the $k$-set cover problem.

Let $(U,\mathscr{S})$ be an instance of the $k$-Set Cover problem, where $U$ is the set of elements to be covered, $\mathscr{S}$ is a collection of sets, and $\bigcup_{S\in \mathscr{S}}S=U$. For $i=k,k-1,...,3$, let $(U_i,\mathscr{S}_i)$ be the instance for phase $i$ of Algorithm PRPSLI, where $U_i$ is the set of elements which have not been covered before Phase $i$ and $\mathscr{S}_i$ is the collection of sets in $S$ which contain only the elements in $U_i$. Let $OPT_i$ be an optimal solution of $(U_i,\mathscr{S}_i)$ for $i\geq 7$. For $i\leq 6$, $OPT_i$ is an optimal solution of $(U_i,\mathscr{S}_i)$ with minimal number of 1-sets. $OPT$ is an optimal solution of $(U,\mathscr{S})$. Let $b_{i,j}$ be the ratio of the number of $j$-sets in $OPT_i$ over the number of sets in $OPT$. Let $\varrho_i$ be the approximation ratio of the set packing algorithm used in Phase $i$. Let $a_1$ be the ratio of the number of $1$-sets chosen by the semi-local optimization phase over the number of sets in $OPT$. Since $|OPT_i|\leq |OPT|$, we have for $i=k,k-1,...,3$,
\begin{equation}
\sum_{j=1}^i b_{i,j}\leq 1 \; .
\end{equation}

In each phase of PRPSLI, the number of $i$-sets chosen by the algorithm is $n_i=\frac{|U_i\backslash U_{i-1}|}{i}$. Since $U_{i-1}\subseteq U_i$, then $|U_i\backslash U_{i-1}|=|U_i|-|U_{i-1}|=(\sum_{j=1}^i jb_{i,j}-\sum_{j=1}^{i-1}jb_{i-1,j})|OPT|$. Let $\varrho_i$ be the approximation ratio of the set packing algorithm used in Phase $i$. At the beginning of Phase $i$, there are $b_{i,i}|OPT|$ $i$-sets. Thus,
\begin{eqnarray}
n_i &=& \frac{(\sum_{j=1}^i jb_{i,j}-\sum_{j=1}^{i-1}jb_{i-1,j})|OPT|}{i} \\
&&{} \geq \varrho_i b_{i,i}|OPT| \; .
\end{eqnarray}
i.e.
\begin{equation}
\sum_{j=1}^{i-1} jb_{i-1,j}-\sum_{j=1}^{i-1}jb_{i,j}-i(1-\varrho_i) b_{i,i}\leq 0 \; .
\end{equation}

We consider additional constraints imposed by the restricted phases, namely for Phase 6 to 3. Let $a_j$ be the ratio of the number of $j$-sets chosen by the semi-local optimization phase over the number of sets in $OPT$, for $j=1,2,3$. In each restricted phase, the number of 1-sets does not increase. Hence, for $i=3,4,5,6$,
\begin{equation}
a_1\leq b_{i,1} \; .
\end{equation}

Next, we obtain an upper bound of the approximation ratio of PRPSLI. From Lemma 2.3 in \cite{furer}, we have $a_1+a_2\leq b_{3,1}+b_{3,2}+b_{3,3}$. Also notice that $a_1+2a_2+3a_3=b_{3,1}+2b_{3,2}+3b_{3,3}$. Thus we have an upper bound of $n_3$, namely,
\begin{eqnarray}
n_3 &=& (a_1+a_2+a_3)|OPT|=(\frac{a_1}{3}+\frac{a_1+a_2}{3}+\frac{a_1+2a_2+3a_3}{3})|OPT| \nonumber \\
&\leq& (\frac{a_1}{3}+\frac{b_{3,1}+b_{3,2}+b_{3,3}}{3}+\frac{b_{3,1}+2b_{3,2}+3b_{3,3}}{3})|OPT| \nonumber \\
&=& (\frac{1}{3}a_1+\frac{2}{3}b_{3,1}+b_{3,2}+\frac{4}{3}b_{3,3})|OPT| \; .
\end{eqnarray}

Combining (9.2) and (9.6), we have an upper bound of the approximation ratio of PRPSLI as,
\begin{eqnarray}
\frac{\sum_{i=3}^k n_i}{|OPT|} &\leq& \sum_{i=4}^k \frac{\sum_{j=1}^i jb_{i,j}-\sum_{j=1}^{i-1}jb_{i-1,j}}{i}+\frac{1}{3}a_1+\frac{2}{3}b_{3,1}+b_{3,2}+\frac{4}{3}b_{3,3} \nonumber \\
&=& \sum_{j=1}^k \frac{j}{k}b_{k,j}+\sum_{i=4}^{k-1}\sum_{j=1}^i \frac{j}{i(i+1)}b_{i,j}\nonumber \\
&&+\frac{1}{3}a_1+\frac{5}{12}b_{3,1}+\frac{1}{2}b_{3,2}+\frac{7}{12}b_{3,3} \; .
\end{eqnarray}

Moreover,
\begin{equation}
b_{i,j}\geq 0 \qquad
\textrm{for } j=1,..,i;\quad i=k,...,3.
\end{equation}
\begin{equation}
a_1\geq 0 \; .
\end{equation}

Hence, we define the factor-revealing linear program of PRPSLI with objective function (9.7) and constraints (9.1), (9.4), (9.5), (9.8), (9.9) as follows,

\[\begin{array}{cll} \max &{\displaystyle \sum_{j=1}^k \frac{j}{k}b_{k,j}+\sum_{i=4}^{k-1}\sum_{j=1}^i \frac{j}{i(i+1)}b_{i,j}+\frac{1}{3}a_1+\frac{5}{12}b_{3,1}+\frac{1}{2}b_{3,2}+\frac{7}{12}b_{3,3}} \\
{\rm s.t.} &{\displaystyle \sum_{j=1}^i b_{i,j}\leq 1,} \quad i=3,...,k,\\
&{\displaystyle \sum_{j=1}^{i-1} jb_{i-1,j}-\sum_{j=1}^{i-1}jb_{i,j}-i(1-\varrho_i) b_{i,i} \leq 0,} \quad i=4,...,k, \\
&{\displaystyle a_1-b_{i,1} \leq 0,} \quad i=3,..,6, \\
&{\displaystyle b_{i,j} \geq 0,} \quad i=3,...,k, j=1,...,i, \\
& a_1 \geq 0. & \textrm{(LP)} \\
\end{array} \]

We also prove that

\begin{lemma}
For any $k\geq 4$, the approximation ratio of Algorithm PRPSLI is upper-bounded by the maximized objective function value of the factor-revealing linear program (LP).
\end{lemma}

\subsection{Finding an Upper bound for the approximation ratio of PRPSLI}

\begin{proof}

Plug $\varrho_i=\frac{2}{i}-\epsilon$ for $i=k,...,5$ and $\varrho_4=\frac{7}{16}$ in (LP). The dual of (LP) is,

\[\begin{array}{cll} \min &{\displaystyle \sum_{i=3}^k \beta_i} \\
{\rm s.t.} & (1) \quad {\displaystyle \delta_3+\delta_4+\delta_5+\delta_6 \geq \frac{1}{3}}, \\
& (2) \quad {\displaystyle \beta_3+\gamma_4-\delta_3 \geq \frac{5}{12}}, \\
& (3) \quad {\displaystyle \beta_3+2\gamma_4 \geq \frac{1}{2}}, \\
& (4) \quad {\displaystyle \beta_3+3\gamma_4 \geq \frac{7}{12}}, \\
& (5) \quad {\displaystyle \beta_i+\gamma_{i+1}-\gamma_i-\delta_i \geq \frac{1}{i(i+1)}}, \quad i=4,5,6, \\
& (6) \quad {\displaystyle \beta_i+j\gamma_{i+1}-j\gamma_i \geq \frac{j}{i(i+1)}}, \quad i=4,...,k-1, j=1,...,i-1, \\
& (7) \quad {\displaystyle \beta_i+i\gamma_{i+1}-2.25\gamma_i \geq \frac{1}{i+1}}, \quad i=4, \\
& (8) \quad {\displaystyle \beta_i+i\gamma_{i+1}-(i-2+\epsilon)\gamma_i \geq \frac{1}{i+1}}, \quad i=5,...,k-1, \\
& (9) \quad {\displaystyle \beta_k-j\gamma_k \geq \frac{j}{k}}, \quad j=1,...,k-1, \textrm{ }k\geq 7 \\
& (9.1) \quad {\displaystyle \beta_k-\gamma_k-\delta_k \geq \frac{1}{k}}, \quad k=4,5,6, \\
& (10) \quad {\displaystyle \beta_k-(k-2+\epsilon)\gamma_k \geq 1}, \\
& (11) \quad {\displaystyle \beta_i \geq 0}, \quad i=3,...,k, \\
& (12) \quad {\displaystyle \gamma_i \geq 0}, \quad i=4,...,k, \\
& (13) \quad {\displaystyle \delta_i \geq 0}, \quad i=3,4,5,6. & \textrm{(Dual)}
\end{array} \]



For $k=4$, set $\gamma_4=\frac{1}{12}$, $\delta_3=0,\delta_4=\frac{1}{3}$, $\beta_3=\frac{1}{3}, \beta_4=1+\frac{1}{12}\cdot\frac{9}{4}$. We have $\sum_{i=3}^4\beta_i=\frac{7}{16}+\frac{1}{12}+1$. \\

For $k=5$, set $\gamma_4=\frac{1}{12},\gamma_5=0$, $\beta_3=\frac{1}{3}, \beta_4=\frac{3}{20}+3\gamma_4, \beta_5=1$, $\delta_3=0,\delta_4=\frac{1}{10}+2\gamma_4,\delta_5=\frac{1}{3}$. We have $\sum_{i=3}^5\beta_i =\frac{2}{5}+\frac{1}{3}+1$. \\

For $k\geq 6$.\\

Set $\delta_3=0$; $\gamma_4=\frac{1}{12}$, $\gamma_i=\gamma_{i+2}+\frac{2}{i(i+1)(i+2)}$ for $i=6,...,k-2$, $\gamma_{k-1}=0$, $\gamma_k=\frac{1}{(k-1)k}$; $\beta_3=\frac{1}{3}$, $\beta_i=\frac{1}{i+1}-i\gamma_{i+1}+(i-2+\epsilon)\gamma_i$ for $i=6,...,k-1$, $\beta_k=1+(k-2+\epsilon)\gamma_k$. \\

For $i\geq 6$, $\gamma_i+\gamma_{i+1}-\frac{1}{i(i+1)}=\gamma_{i+1}+\gamma_{i+2}+\frac{2}{i(i+1)(i+2)}-\frac{1}{i(i+1)}=\gamma_{i+1}+\gamma_{i+2}-\frac{1}{(i+1)(i+2)}$. By induction, $\gamma_i+\gamma_{i+1}-\frac{1}{i(i+1)}=\gamma_{k-1}+\gamma_{k}-\frac{1}{(k-1)k}=0$. Hence,
\begin{equation}
\gamma_i+\gamma_{i+1}=\frac{1}{i(i+1)}, \quad \textrm{for }i=6,...,k-1. \nonumber
\end{equation}

Then, constraints (2),(3) and (4) hold as equality. \\

In constraint (6) for $i\geq 6$, $j\leq i-1$, $\beta_i+j\gamma_{i+1}-j\gamma_i=\frac{1}{i+1}-i\gamma_{i+1}+(i-2+\epsilon)\gamma_i+j\gamma_{i+1}-j\gamma_i=\frac{1}{i+1}-(i-j)\gamma_{i+1}+(i-2-j+\epsilon)(\frac{1}{i(i+1)}-\gamma_{i+1})
=\frac{1}{i+1}(1+\frac{i-j-2+\epsilon}{i})-(2i-2-2j+\epsilon)\gamma_{i+1}\geq\frac{1}{i+1}+\frac{i-2-j+\epsilon}{i(i+1)}-\frac{2i-2-2j+\epsilon}{i(i+1)}=\frac{j}{i(i+1)}$. \\

In constraint (8) for $i\geq 6$, inequalities hold as equality. \\

Constraint (9) holds, $\beta_k-j\gamma_k=1+(k-2)\gamma_k-j\gamma_k\geq 1-\gamma_k\geq \frac{j}{k}$. \\

Constraint (10) holds as equality. \\

Set $\gamma_5=\frac{1}{30}-\gamma_6$, $\delta_4=2\gamma_4-2\gamma_5+\frac{1}{10}$, $\delta_5=2\gamma_5-4\gamma_6+\frac{4}{30}$, $\delta_6=0$ for odd $k$ and $\delta_6=\frac{1}{15}$ for even $k$. \\

$\beta_4\geq\max\{\frac{1}{5}+\frac{9}{4}\gamma_4-4\gamma_5,\frac{1}{20}+\gamma_4-\gamma_5+\delta_4,\frac{3}{20}+3\gamma_4-3\gamma_5\}=\frac{3}{20}+3\gamma_4-3\gamma_5$, $\beta_5\geq\max\{\frac{1}{6}+(3+\epsilon)\gamma_5-5\gamma_6,\frac{1}{30}+\gamma_5-\gamma_6+\delta_5,\frac{4}{30}+4\gamma_5-4\gamma_6\}=\frac{1}{6}+(3+\epsilon)\gamma_5-5\gamma_6$. Let $\beta_4=\frac{3}{20}+3\gamma_4-3\gamma_5$, $\beta_5=\frac{1}{6}+(3+\epsilon)\gamma_5-5\gamma_6$. \\

For $k$ is odd and $k\geq 9$, $\gamma_6=\gamma_6-\gamma_8+\cdots+\gamma_{k-3}-\gamma_{k-1}+\gamma_{k-1}=2(\frac{1}{6\cdot 7\cdot 8}+\cdots+\frac{1}{(k-3)(k-2)(k-1)})$ increases monotonically with respect to $k$. Thus, for any odd $k$ and $k\geq 9$,  $\gamma_6=2H_{\frac{k-3}{2}}-2H_{k-2}+\frac{1}{k-1}+\frac{7}{5}\leq -2\ln 2+\frac{7}{5}<\frac{1}{60}$.\\
For $k=7$, $\gamma_6=0$. \\

Hence, Constraint (1) holds, $\delta_3+\delta_4+\delta_5+\delta_6=2\gamma_4+\frac{1}{10}+\frac{4}{30}-4\gamma_6=\frac{2}{5}-4\gamma_6>\frac{1}{3}$. \\

For $k$ is even and $k\geq 6$, $\gamma_6=\gamma_6-\gamma_8+\cdots+\gamma_{k-2}-\gamma_{k}+\gamma_{k}=2(\frac{1}{6\cdot 7\cdot 8}+\cdots+\frac{1}{(k-2)(k-1)k})+\frac{1}{(k-1)k}=2H_{\frac{k-2}{2}}-2H_{k-1}+\frac{1}{k-1}+\frac{7}{5}$ decreases monotonically with respect to $k$.
For $k=6$, $\gamma_6=\frac{1}{30}$. Thus, for any even $k$ and $k\geq 6$, $\gamma_6\leq \frac{1}{30}$. Moreover, $\gamma_6>-2\ln 2+\frac{7}{5}>\frac{1}{60}$

Hence, Constraint (1) holds, $\delta_3+\delta_4+\delta_5+\delta_6=2\gamma_4+\frac{1}{10}+\frac{4}{30}-4\gamma_6+\frac{1}{15}=\frac{7}{15}-4\gamma_6\geq\frac{1}{3}$. \\

Constraint (5) for $i=4,5$, constraint (6) for $i=4,5$, constraint (7) and constraint (8) for $i=5$ hold directly as a result of the settings of these parameters. \\

Constraint (5) for $i=6$ holds, $\beta_6+\gamma_7-\gamma_6-\delta_6-\frac{1}{42}=\frac{1}{7}+3\gamma_6-5\gamma_7-\frac{1}{15}=\frac{1}{7}+3\gamma_6-5(\frac{1}{42}-\gamma_6)-\frac{1}{15}-\frac{1}{42}=8\gamma_6-\frac{1}{15}>\frac{8}{60}-\frac{1}{15}>0$.\\

Constraint (9.1) holds for $k=6$, $\beta_6-\gamma_6-\delta_6>1+3\gamma_6-\frac{1}{15}>\frac{1}{42}$. \\

Moreover, constraint (11), (12), (13) hold. \\

Finally, we compute the value of the objective function. \\

For odd $k$ and $k\geq 7$,
$\sum_{i=3}^k\beta_i=\beta_3+\beta_4+\beta_5+\sum_{i=6}^k\beta_i=\frac{1}{3}+\frac{3}{20}+\frac{1}{6}+3\gamma_4+\epsilon\gamma_5-5\gamma_6
+\sum_{i=6}^{k-1}\frac{1}{i+1}-i\gamma_{i+1}+(i-2+\epsilon)\gamma_i+1+(k-2+\epsilon)\gamma_k
=\frac{1}{3}+\frac{3}{20}+\frac{1}{6}+3\gamma_4-5\gamma_6+1+\frac{1}{7}-\gamma_7+4\gamma_6+\sum_{i=8}^k \frac{1}{i}-(\frac{1}{8\cdot 9}+\frac{1}{10\cdot 11}+\cdots+\frac{1}{(k-1)k})+\epsilon\sum_{i=5}^k \gamma_i=1+\frac{1}{3}+\frac{2}{5}+\frac{1}{6}+\frac{1}{7}-5\gamma_6-(\frac{1}{42}-\gamma_6)+4\gamma_6+2(\frac{1}{9}+\cdots\frac{1}{k})+\epsilon\sum_{i=5}^k\gamma_i
\leq 1+\frac{1}{3}+\frac{2}{5}+\frac{2}{7}+\cdots+\frac{2}{k}+\epsilon$. \\

Last inequality holds because $\sum_{i=5}^k\gamma_i\leq\sum_{i=5}^k\frac{1}{(i+1)i}\leq\frac{1}{5}$. \\

Similarly, for even $k$ and $k\geq 6$, $\sum_{i=3}^k\beta_i=\beta_3+\beta_4+\beta_5+\sum_{i=6}^k\beta_i=\frac{1}{3}+\frac{3}{20}+\frac{1}{6}+3\gamma_4+\epsilon\gamma_5-5\gamma_6
+\sum_{i=6}^{k-1}\frac{1}{i+1}-i\gamma_{i+1}+(i-2+\epsilon)\gamma_i+1+(k-2+\epsilon)\gamma_k
=\frac{1}{3}+\frac{3}{20}+\frac{1}{6}+3\gamma_4-5\gamma_6+1+\frac{1}{7}-\gamma_7+4\gamma_6+\sum_{i=8}^k \frac{1}{i}-(\frac{1}{8\cdot 9}+\frac{1}{10\cdot 11}+\cdots+\frac{1}{(k-2)(k-1)}+\frac{1}{(k-1)k})+\epsilon\sum_{i=5}^k \gamma_i=1+\frac{1}{3}+\frac{2}{5}+\frac{1}{6}+\frac{1}{7}-5\gamma_6-(\frac{1}{42}-\gamma_6)+4\gamma_6+2(\frac{1}{9}+\cdots\frac{1}{k-3})+\frac{1}{k-1}+\frac{2}{k}+\epsilon\sum_{i=5}^k\gamma_i
\leq 1+\frac{1}{3}+\frac{2}{5}+\frac{2}{7}+\cdots+\frac{2}{k-3}+\frac{1}{k-1}+\frac{2}{k}+\epsilon$. \\

Therefore, the approximation ratio of PRPSLI for odd $k$ and $k\geq 7$ can be upper bounded by $1+\frac{1}{3}+\frac{2}{5}+\frac{2}{7}+\cdots+\frac{2}{k}+\epsilon$. For even $k$ and $k\geq 6$, it is upper bounded by $1+\frac{1}{3}+\frac{2}{5}+\frac{2}{7}+\cdots+\frac{2}{k-3}+\frac{1}{k-1}+\frac{2}{k}+\epsilon$. \qed
\end{proof}

\subsection{Tight example of PRPSLI}
For every $k\geq 4$ and any $\epsilon>0$, we give a tight example of PRPSLI based on the tight example of the semi-local $(2,1)$-improvement \cite{furer} for 3-Set Cover, the tight example of the Restricted 4-Set Packing algorithm we give in Appendix Section 7.2, and the tight example of the Restricted $k$-Set Packing algorithm for $k\geq 5$, which is the same as the tight example of the $k$-set packing heuristic \cite{schrijver}. \\

We assume that the optimal solution $\mathscr{O}$ consists of only disjoint $k$-sets. To calculate the performance ratio on this instance, we charge a cost of 1 for each set chosen by the algorithm, and the cost is uniformly distributed to every element of the chosen set \cite{furer}.

\begin{itemize}

    \item \textbf{$k=4$}. In Phase 4, the Restricted 4-Set Packing algorithm covers 1 element of each set in a $\frac{1}{4}(1-\epsilon)$ fraction of $\mathscr{O}$, 2 elements of each set in a $\frac{3}{4}(1-\epsilon)$ fraction, and 3 elements of each set in the remaining $\epsilon$ fraction. Denote the three parts of $\mathscr{O}$ by $\mathscr{O}_1$, $\mathscr{O}_2$ and $\mathscr{O}_3$ respectively. In Phase 3, the semi-local optimization covers 1 element in each set of $\mathscr{O}_1$ by 3-sets, and the remaining uncovered elements of $\mathscr{O}$ are covered by 2-sets. The performance ratio of PRPSLI on this instance is $\frac{3}{4}(1-\epsilon)(\frac{1}{2}+1)+\frac{1}{4}(1-\epsilon)(\frac{1}{4}+\frac{1}{3}+1)+\epsilon(\frac{3}{4}+\frac{1}{2})=\frac{7}{16}+\frac{1}{12}+1-\frac{13}{48}\epsilon$. \\

    \item \textbf{$k=5$}. In Phase 5, the Restricted 5-Set Packing algorithm covers 2 elements of each set in a $1-\epsilon$ fraction of $\mathscr{O}$, 1 element of each set in the remaining $\epsilon$ fraction. Denote the two parts of $\mathscr{O}$ by $\mathscr{O}_2$ and $\mathscr{O}_1$ respectively. The algorithm switches to 4-Set Cover on $\mathscr{O}_1$ and it performs a semi-local optimization on $\mathscr{O}_2$. The performance ratio of Algorithm 1 on this instance is $(1-\epsilon)(\frac{2}{5}+\frac{1}{3}+1)+\epsilon(\frac{1}{5}+\frac{7}{16}+\frac{1}{12}+1)=\frac{2}{5}+\frac{1}{3}+1-\frac{1}{80}\epsilon$. \\

    \item \textbf{$k$ odd and $k\geq 7$}. In Phase $k$, the $k$-Set Packing algorithm covers 2 elements of each set in a $1-\epsilon$ fraction of $\mathscr{O}$, 1 element of each set in the remaining $\epsilon$ fraction. Denote the two parts of $\mathscr{O}$ by $\mathscr{O}_2$ and $\mathscr{O}_1$ respectively. In Phase $k-1$, the algorithm covers 1 element of each set in $\mathscr{O}_1$. Then it switches to $(k-2)$-Set Cover on the remaining uncovered elements. The performance ratio of PRPSLI on this instance is at least $(1-\epsilon)\frac{2}{k}+\epsilon(\frac{1}{k}+\frac{1}{k-1})+\rho_{k-2}$, i.e. $\frac{2}{k}+\rho_{k-2}+(\frac{1}{k-1}-\frac{1}{k})\epsilon$, which is $\frac{2}{7}+\frac{2}{5}+\frac{1}{3}+1+\frac{19}{1680}\epsilon$ for $k=7$, and by induction  $\frac{2}{k}+\frac{2}{k-2}+\cdots+\frac{2}{5}+\frac{1}{3}+1+(\frac{1}{k-1}-\frac{1}{k}+\frac{1}{k-3}-\frac{1}{k-2}+\cdots+\frac{1}{8}-\frac{1}{9}+\frac{19}{1680})\epsilon$ for $k\geq 9$. The coefficient of $\epsilon$ is upper bounded by $\frac{1}{k-1}-\frac{1}{k}+\frac{1}{k-3}-\frac{1}{k-1}+\cdots+\frac{1}{8}-\frac{1}{10}+\frac{19}{1680}= \frac{1}{8}-\frac{1}{k}+\frac{19}{1680}\leq 1$. Hence, the performance ratio is at most $\frac{2}{k}+\frac{2}{k-2}+\cdots+\frac{2}{5}+\frac{1}{3}+1+\epsilon$ for $k\geq 9$. \\

    \item \textbf{$k$ even and $k\geq 6$}. In Phase $k$, the $k$-Set Packing algorithm covers 2 elements of each set in a $1-\epsilon$ fraction of $\mathscr{O}$, 1 elements of each set in the remaining $\epsilon$ fraction. Denote the two parts of $\mathscr{O}$ by $\mathscr{O}_2$ and $\mathscr{O}_1$ respectively. In Phase $k-1$, the $(k-1)$-Set Packing algorithm covers 1 element of each set in $\mathscr{O}_1$ and then 1 element of each set in $\mathscr{O}$. Then the algorithm switches to $(k-3)$-Set Cover on the remaining uncovered elements. The performance ratio of Algorithm 1 on this instance is at least $(1-\epsilon)\frac{2}{k}+\epsilon(\frac{1}{k}+\frac{1}{k-1})+\frac{1}{k-1}+\rho_{k-3}$, i.e. $\frac{2}{k}+\frac{1}{k-1}+\rho_{k-3}+(\frac{1}{k-1}-\frac{1}{k})\epsilon$, which is $\frac{2}{6}+\frac{1}{5}+\frac{1}{3}+1+\frac{1}{30}\epsilon$ for $k=6$ and by a similar argument as the above case, at most $\frac{2}{k}+\frac{1}{k-1}+\frac{2}{k-3}+\cdots+\frac{2}{5}+\frac{1}{3}+1+\epsilon$ for $k\geq 8$. \\

\end{itemize}


\end{document}